\documentclass[journal, twocolumn]{IEEEtran}

\usepackage{times}

\usepackage{amsmath}  
\usepackage{amssymb}  
\usepackage{mathrsfs} 

\usepackage{theorem}  
\usepackage{cite}     
\usepackage{comment}  

\usepackage{upref}
\usepackage{amsfonts}

\usepackage{verbatim}

\usepackage[dvipsnames,usenames]{color}

\usepackage{graphicx}
\usepackage{subfigure}

\newcommand{\be}[1]{\begin{equation}\label{#1}}
\newcommand{\ee}{\end{equation}}

\newcommand{\bc}{\begin{center}}
\newcommand{\ec}{\end{center}}

\newcommand{\cC}{{\cal C}}
\newcommand{\cD}{{\cal D}}
\newcommand{\cE}{{\cal E}}

\newcommand{\cS}{{\cal S}}

\newcommand{\bfb}{{\boldsymbol b}}
\newcommand{\bfc}{{\boldsymbol c}}
\newcommand{\bfd}{{\boldsymbol d}}

\newcommand{\bfg}{{\boldsymbol g}}

\newcommand{\bfq}{{\boldsymbol q}}

\newcommand{\bfs}{{\boldsymbol s}}
\newcommand{\bft}{{\boldsymbol t}}
\newcommand{\bfu}{{\boldsymbol u}}
\newcommand{\bfv}{{\boldsymbol v}}
\newcommand{\bfw}{{\boldsymbol w}}
\newcommand{\bfx}{{\boldsymbol x}}
\newcommand{\bfy}{{\boldsymbol y}}
\newcommand{\bfz}{{\boldsymbol z}}

\newcommand{\BFX}{{\boldsymbol X}}

\newcommand{\bfX}{{\mathbf X}}

\newcommand{\bfZ}{{\mathbf Z}}

\renewcommand{\leq}{\leqslant}

\renewcommand{\geq}{\geqslant}

\newcommand{\Cref}[1]{Co\-rol\-la\-ry\,\ref{#1}}

\theoremstyle{plain} \theorembodyfont{\normalfont\slshape}

\newtheorem{thm}{Theorem$\!$}
\newenvironment{theorem}{\begin{thm}\hspace*{-1ex}{\bf.}}{\end{thm}}

\newtheorem{prop}{Proposition$\!$}

\newtheorem{lem}{Lemma$\!$}
\newenvironment{lemma}{\begin{lem}\hspace*{-1ex}{\bf.}}{\end{lem}}

\newtheorem{cor}{Corollary$\!$}
\newenvironment{corollary}{\begin{cor}\hspace*{-1ex}{\bf.}}{\end{cor}}

\newtheorem{cl}{Claim$\!$}

\newtheorem{defi}{Definition$\!$}
\newenvironment{definition}{\begin{defi}\hspace*{-1ex}{\bf .}}{\end{defi}}

\newtheorem{const}{Construction$\!$}

\newtheorem{algr}{Algorithm$\!$}

\theorembodyfont{\normalfont}

\newtheorem{exam}{Example$\!$}
\newenvironment{example}{\begin{exam}\hspace*{-1ex}{\bf .}}{\end{exam}}

\newtheorem{remrk}{Remark$\!$}
\newenvironment{remark}{\begin{remrk}\hspace*{-1ex}{\bf .}}{\end{remrk}}

\newlength{\paragraphindent}
\newlength{\widthone}
\newlength{\widthtwo}
\newlength{\widththree}
\newlength{\colwidthtemp}
\definecolor{Codecolor}{named}{White}  

\newcommand{\Copen}{\mbox{\{\kern-5.50pt\{}}
\newcommand{\Cclose}{\mbox{\}\kern-5.50pt\}}}
\newcommand{\Cslash}{\mbox{$\backslash\kern-6.02pt\backslash$}}

\begin{document}

\title{$\,$
{Time-Space Constrained Codes for Phase-Change Memories}}
\author{\large Minghai~Qin,~\IEEEmembership{Student Member,~IEEE,} Eitan~Yaakobi,~\IEEEmembership{Member,~IEEE,} Paul~H.~Siegel,~\IEEEmembership{Fellow,~IEEE,}

\thanks{
This work was presented in part at the {\em IEEE Global Telecommunications Conference (GLOBECOM 2011)}, Texas, USA, (December 2011), pp.\ 1-6  in the proceedings.

M.\ Qin, and P.\ H.\ Siegel are with the Department of Electrical and Computer Engineering and the Center for Magnetic Recording Research, University of California at San Diego, La Jolla, CA 92093, U.S.A. (e-mail: \{mqin,eyaakobi,psiegel\}@ucsd.edu).

E. Yaakobi is with the Department of Electrical Engineering, California Institute of Technology, Pasadena, CA 91125 USA, and also with the Department of Electrical and Computer Engineering and the Center for Magnetic Recording Research, University of California, San Diego, La Jolla, CA 92093 USA (e-mail: eyaakobi@ucsd.edu).

This research was supported in part by the ISEF Foundation, the Lester Deutsch Fellowship, the University of California Lab Fees Research Program, Award No. 09-LR-06-118620-SIEP, the National Science Foundation under Grant  CCF-1116739, and the Center for Magnetic Recording Research at the University of California, San Diego.
}
}
\maketitle

\begin{abstract}
Phase-change memory (PCM) is a promising non-volatile solid-state memory technology. A PCM cell stores data by using its amorphous and crystalline states. The cell changes between these two states using high temperature. However, since the cells are sensitive to high temperature, it is important, when programming cells, to balance the heat both in time and space.

In this paper, we study the time-space constraint for PCM, which was originally proposed by Jiang et al. A code is called an \emph{$(\alpha,\beta,p)$-constrained code} if for any $\alpha$ consecutive rewrites and for any segment of $\beta$ contiguous cells, the total rewrite cost of the $\beta$ cells over those $\alpha$ rewrites is at most $p$. Here, the cells are binary and the rewrite cost is defined to be the Hamming distance between the current and next memory states. First, we show a general upper bound on the achievable rate of these codes which extends the results of Jiang et al. Then, we generalize their construction for $(\alpha\geq 1, \beta=1,p=1)$-constrained codes and show another construction for $(\alpha = 1, \beta\geq 1,p\geq1)$-constrained codes. Finally, we show that these two constructions can be used to construct codes for all values of $\alpha$, $\beta$, and $p$.
\end{abstract}

\section{Introduction}\label{sec:introduction}\noindent
Phase-change memory (PCM) devices are a promising technology for non-volatile memories. Like a flash memory, a PCM consists of cells that can be in distinct physical states. In the simplest case, the PCM cell has two possible states, an amorphous state and a crystalline state. Multiple-bit per cell PCMs can be implemented by using partially crystalline states~\cite{B_etal10}.

While in a flash memory one can decrease a cell level only by erasing the entire block of about $10^6$ cells that contains it, in a PCM one can independently decrease an individual cell level -- but only to level zero. This operation is called a RESET operation. A SET operation can then be used to change the cell state to any valid level. Therefore, in order to decrease a cell level from one non-zero value to a smaller non-zero value, one needs to first RESET the cell to level zero, and then SET it to the new desired level~\cite{B_etal10}. Thus, as with flash memory programming, there is a significant asymmetry between the two operations of increasing and decreasing a cell level.

As in a flash memory, a PCM cell has a limited lifetime; the cells can tolerate only about $10^7 - 10^8$ RESET operations before beginning to degrade~\cite{FW08}. Therefore, it is still important when programming cells to minimize the number of RESET operations. Furthermore, a RESET operation can negatively affect the performance of a PCM in other ways. One of them is due to the phenomenon of thermal crosstalk. When a cell is RESET, the levels of its adjacent cells may inadvertently be increased due to heat diffusion associated with the operation~\cite{B_etal10,P_etal04}. Another problem, called thermal accumulation, arises when a small area is subjected to a large number of program operations over a short period of time~\cite{B_etal10,P_etal04}. The resulting accumulation of heat can significantly limit the minimum write latency of a PCM, since the programming accuracy is sensitive to temperature. It is therefore desirable to balance the thermal accumulation over a local area of PCM cells in a fixed period of time. Coding schemes can help overcome the performance degradation resulting from these physical phenomena. Lastras et al.~\cite{L_etal09} studied the capacity of a Write-Efficient Memory (WEM)~\cite{AZ89} for a cost function that is associated with the write model of phase-change memories described above.


Jiang et al.~\cite{JBL10} have proposed codes to mitigate thermal cross-talk and heat accumulation effects in PCM. Under their thermal cross-talk model, when a cell is RESET, the levels of its immediately adjacent cells may also be increased. Hence, if these neighboring cells exceed their target level, they also will have to be RESET, and this effect can then propagate to many more cells. In~\cite{JBL10}, they considered a special case of this and proposed the use of constrained codes to limit the propagation effect. Capacity calculations for these codes were also presented.

The other problem addressed in~\cite{JBL10} is that of heat accumulation. In this model, the {\em rewrite cost} is defined to be the number of programmed cells, i.e., the Hamming distance between the current and next cell-state vectors. A code is said to be \emph{$(\alpha,\beta,p)$-constrained} if for any $\alpha$ consecutive rewrites and for any segment of $\beta$ contiguous cells, the total rewrite cost of the $\beta$ cells over those $\alpha$ rewrites is at most $p$. A specific code construction was given for the $(\alpha\geq 1, \beta =1, p=1)$-constraint as well as an upper bound on the achievable rate of codes for this constraint. An upper bound on the achievable rate was also given for $(\alpha= 1, \beta \geq 1, p=1)$-constrained codes.

The work in~\cite{JBL10} dealt with only a few instances of the parameters $\alpha,\beta$ and $p$. In this paper, we extend the code constructions and achievable-rate bounds to a larger portion of the parameter space. In Section~\ref{sec:preliminaries}, we formally define the constrained-coding problem for PCM. In Section~\ref{sec:upper}, using connections to two-dimensional constrained coding, we present a scheme to calculate an upper bound on the achievable rate for all values of $\alpha,\beta$ and $p$. If the value of $\alpha$ or $\beta$ is 1 then the two-dimensional constraint becomes a one-dimensional constraint and we calculate the upper bound on the achievable rate for all values of $p$. This result coincides with the result in~\cite{JBL10} for $(\alpha\geq 1, \beta =1, p=1)$ and $(\alpha= 1, \beta \geq 1, p=1)$. We also derive upper bounds for some cases with parameters satisfying $(\alpha> 1, \beta > 1, p=1)$ using known results on the upper bound of two-dimensional constrained codes. In Section~\ref{sec:lower bound}, code constructions are given. First, a trivial construction is given and we show an improvement for $(\alpha = 1,\beta\geq1,p\geq1)$-constrained codes and extend the construction in~\cite{JBL10} of $(\alpha\geq 1,\beta=1,p=1)$-constrained codes to arbitrary $p$. Finally, we show how to extend the constructions for all values of $\alpha,\beta$ and $p$.

\section{Preliminaries}\label{sec:preliminaries}\noindent
In this section, we give a formal definition of the constrained-coding problem. The number of cells is denoted by $n$ and the memory cells are binary. The cell-state vectors are the binary vectors from $\{0,1\}^n$. If a cell-state vector $\bfu=(u_1,\ldots,u_n)\in\{0,1\}^n$ is rewritten to another cell-state vector $\bfv=(v_1,\ldots,v_n)\in\{0,1\}^n$, then the rewrite cost is defined to be the Hamming distance between $\bfu$ and $\bfv$, that is
$$ d_H(\bfu,\bfv) = |\{i\ : \ u_i\neq v_i, 1\leq i\leq n\}|.$$
The Hamming weight of a vector $\bfu$ is $wt(\bfu) = d_H(\bfu,\textbf{0})$. The complement of a vector $\bfu$ is
$\overline{\bfu}=(\overline{u}_1,\ldots,\overline{u}_n)$. For a vector $\bfx=(x_1,\ldots,x_n)$, we define for all $1\leq p\leq q\leq n$, $\bfx_p^q=(x_p,x_{p+1},\ldots,x_q)$. The set $\{i,i+1,\ldots,j\}$ is denoted by $[i:j]$ for $i\leq j$, and in particular, $\{1,2,\ldots,\lfloor 2^{nR} \rfloor\}$ is denoted by $[1:2^{nR}]$ for an integer $n$ and real $R$.

\begin{definition}
Let $\alpha,\beta,p$ be positive integers. A code $\cC$ satisfies the \textbf{$\boldsymbol{(\alpha,\beta,p)}$-constraint} if for any $\alpha$ consecutive rewrites and for any segment of $\beta$ contiguous cells, the total rewrite cost of those $\beta$ cells over those $\alpha$ rewrites is at most $p$, and $\cC$ is called an \textbf{$\boldsymbol{(\alpha,\beta,p)}$-constrained code}. That is, if $\bfv_i= (v_{i,1},\ldots,v_{i,n})$, for $i\geq 1$, is the cell-state vector on the $i$-th write, then, for all $i\geq 1$ and $1\leq j\leq n-\beta+1$,
\[
\Big| \{(k,\ell):v_{i+k,j+\ell}\neq v_{i+k+1,j+\ell}, 0\leq k < \alpha, 0\leq \ell<\beta\}\Big|\leq p.
\]
\end{definition}
We will specify $(\alpha,\beta,p)$-constrained codes by an explicit construction of their encoding and decoding maps. On the $i$-th write, for $i\geq 1$, the encoder
\[
\cE_i: [1:2^{nR_i}] \times\{0,1\}^n\mapsto \{0,1\}^n
\]
maps the new information symbol and the current cell-state vector to the next cell-state vector. The decoder
\[
\cD_i:\{0,1\}^n \mapsto [1:2^{nR_i}]
\]
maps the cell-state vector to the represented information symbol. We denote the \emph{individual rate} on the $i$-th write of the $(\alpha,\beta,p)$-constrained code by $R_i$. Note that the alphabet size of the messages on each write does not have to be the same. The \emph{rate} $R$ of the $(\alpha,\beta,p)$-constrained code is defined as

\begin{equation}\label{eq:rate2}
R= \lim_{m\rightarrow \infty}\frac{\sum_{i=1}^mR_i}{m}.
\end{equation}

We assume that the number of writes is large and in the constructions we present there will be a periodic sequence of writes. Thus, it will be possible to change any $(\alpha,\beta,p)$-constrained code $\cC$ with varying individual rates to an $(\alpha,\beta,p)$-constrained code $\cC'$ with fixed individual rates such that the rates of the two constrained codes are the same. This can be achieved by using multiple copies of the code $\cC$ and in each copy of $\cC$ to start writing from a different write within the period of writes. Therefore, we assume that there is no distinction between the two cases and the rate is as defined in Equation~(\ref{eq:rate2}), which is the number of bits written per cell per write.

The encoding and decoding maps can be either the same on all writes or can vary among the writes. In the latter case, we will need more cells in order to index the write number. However, arguing as in~\cite{YKSVW10}, it is possible to show that these extra cells do not reduce asymptotically the rate and therefore we assume here that the encoder and decoder know the write number.

A rate $R$ is called an \textbf{\emph{$\boldsymbol{(\alpha,\beta,p)}$-achievable rate}} if there exists an $(\alpha,\beta,p)$-constrained code $\cC$ such that the rate of $\cC$ is $R$. We denote by $C_{n}(\alpha,\beta,p)$ the supremum of all $(\alpha,\beta,p)$-achievable rates while fixing the number of cells to be $n$. The \textbf{\emph{$\boldsymbol{(\alpha,\beta,p)}$-capacity}} of the $(\alpha,\beta,p)$-constraint is denoted by $\boldsymbol{C(\alpha,\beta,p)}$ and is defined to be
$$C(\alpha,\beta,p) = \lim_{n\rightarrow \infty}C_{n}(\alpha,\beta,p).$$

Our goal in this paper is to give lower and upper bounds on the $(\alpha,\beta,p)$-capacity, $C(\alpha,\beta,p)$, for all values of $\alpha,\beta,$ and $p$. Clearly, if $p\geq\alpha\beta$ then $C(\alpha,\beta,p)=1$. So we assume throughout the paper that $p<\alpha\beta$. Lower bounds will be given by specific constrained code constructions while the upper bounds will be derived analytically using tools drawn from the theory of one- and two-dimensional constrained codes.

\section{Upper Bound on the Capacity}\label{sec:upper}\noindent
In this section, we will present upper bounds on the $(\alpha,\beta,p)$-capacity obtained using techniques from the analysis of two-dimensional constrained codes. There are a number of two-dimensional constraints that have been extensively studied, e.g., 2-dimensional $(d,k)$-runlength-limited (RLL)~\cite{KZ99,SW92}, no isolated bits~\cite{HCRSW04,FL06}, and the checkerboard constraint~\cite{WB98,NZ03}. Given a two-dimensional constraint $S$, its capacity is defined to be
\[
C_{2D}(S) = \lim_{m,n\rightarrow\infty}\frac{\log_2c_S(m,n)}{mn},
\]
where $c_S(m,n)$ is the number of $m\times n$ arrays that satisfy the constraint $S$. The constraint of interest for us in this work is the one where in each rectangle of size $a\times b$, the number of ones is at most $p$.

\begin{definition}
Let $a,b,p$ be positive integers. An $(m\times n)$-array $A=(a_{i,j})_{1\leq i\leq m,1\leq j\leq n}\in\{0,1\}^{m\times n}$ is called an {\em $(a,b,p)$-array} if in each sub-array of $A$ of size $a\times b$, the number of 1's is at most $p$. That is, for all $1\leq i\leq m-a+1$, $1\leq j\leq n-b+1$,
\[
\big|\{(k,\ell)\ :\ 0\leq k\leq a-1,0\leq \ell\leq b-1, a_{i+k,j+\ell}=1\}\big| \leq p.
\]
The capacity of the constraint is denoted by $C_{2D}(a,b,p)$.
\end{definition}

Note that when $p=1$, the $(a,a,1)$ constraint coincides with the square checkerboard constraint of order $a-1$~\cite{WB98}.

The connection between the capacity of the two-dimensional constraint $C_{2D}(a,b,p)$ and the $(\alpha,\beta,p)$-capacity is the following.
\begin{theorem}\label{th:upper}
For all $\alpha,\beta,p$, $C(\alpha,\beta,p)\leq C_{2D}(\alpha,\beta,p)$.
\end{theorem}

\begin{proof}
Let $\cC$ be an $(\alpha,\beta,p)$-constrained code of length $n$. For any sequence of $m$ writes, let us denote by $\bfv_i$, for $i\geq 0$, the cell-state vector on the $i$-th write, where $\bfv_0$ is the all-zero vector. The $m\times n$-array $A=(a_{i,j})$ is defined to be
\[
a_{i,j} = v_{i,j}+v_{i-1,j},
\]
where the addition is a modulo 2 sum. That is, $a_{i,j} = 1$ if and only if the $j$-th cell is changed on the $i$-th write. Since $\cC$ is an $(\alpha,\beta,p)$-constrained code, for all $1\leq i\leq m-\alpha $ and $1\leq j\leq n-\beta+1$,
\[
\big| \{(k,\ell):v_{i+k,j+\ell}\neq v_{i+k+1,j+\ell}, 0\leq k < \alpha, 0\leq \ell<\beta\}\big|\leq p,
\]
and therefore
\[
\big|\{(k,\ell)\ :\ 0\leq k\leq \alpha-1,0\leq \ell\leq \beta-1, a_{i+k,j+\ell}=1\}\big| \leq p.
\]
 Thus, $A$ is an $(\alpha,\beta,p)$-array of size $m\times n$.

Every write sequence of the code $\cC$ corresponds to an $(\alpha,\beta,p)$-array and thus the number of write sequences of length $m$ is at most the number of $(\alpha,\beta,p)$-arrays, which is upper bounded by $2^{mnC_{2D}(\alpha,\beta,p)}$, for $m,n$ large enough. Hence, the number of distinct write sequences is at most $2^{mnC_{2D}(\alpha,\beta,p)}$. However, if the individual rate on the $i$-th write is $R_i$, then the total number of distinct write sequences is $\prod_{i=1}^{m}2^{nR_i}$. We conclude that
\[
\prod_{i=1}^{m}2^{nR_i} \leq 2^{mnC_{2D}(\alpha,\beta,p)}
\]
and, therefore,
\[
\frac{\sum_{i=1}^mR_i}{m} \leq C_{2D}(\alpha,\beta,p).
\]
If $m$ goes to infinity, the rate of any $(\alpha,\beta,p)$-constrained code $R$ satisfies
\[
R\leq C_{2D}(\alpha,\beta,p),
\]
i.e., $C(\alpha,\beta,p)\leq C_{2D}(\alpha,\beta,p)$.
\end{proof}

Theorem~\ref{th:upper} provides a scheme to calculate an upper bound on the $(\alpha,\beta,p)$-capacity from an upper bound on the capacity of a two-dimensional rectangular checkerboard constraint. Unfortunately, good upper bounds are known only for some special cases of the values of $\alpha,\beta,p,$ and in particular, when $p=1$. The checkerboard constraint has attracted considerable attention over the past 20 years and some lower and upper bounds on the capacity were given in~\cite{WB98,TR11,NZ03}. For instance, some upper bounds for the square checkerboard constraint are shown in~\cite{WB98}, from which we can conclude that $C(2,2,1)\leq 0.43431$ and $C(3,3,1)\leq 0.25681$.

In the rest of this section we discuss the cases where $\alpha=1$ or $\beta=1$ corresponding to one-dimensional constraints. First, let us consider the upper bound on $C(\alpha=1,\beta,p)$. We use the one-dimensional $(d,k)$-runlength-limited (RLL) constrained codes~\cite{ZW88}, where the number of 0's between adjacent 1's is at least $d$ and at most $k$. In fact, Jiang et al.~\cite{JBL10} showed that the capacity of the $(\beta-1,\infty)$-RLL constraint is an upper bound on $C(1,\beta,1)$, which is a special case of Theorem~\ref{th:upper}. The lowest curve in Fig.~\ref{figure:upper-space} shows the capacity of the $(\beta-1,\infty)$-RLL constraint. We extend the upper bounds to arbitrary $p\geq1$. First, let us generalize the definition of RLL-constrained codes.

\begin{definition}
Let $\beta,p$ be two positive integers. A binary vector $\bfu$ satisfies the \textbf{$\boldsymbol{(\beta,p)}$-window-weight-limited (WWL) constraint} if for any $\beta$ consecutive cells there are at most $p$ 1's and $\bfu$ is called a {\em $(\beta,p)$-WWL vector}. We denote the capacity of the constraint by $C_{WWL}(\beta,p)$.
\end{definition}

Note that for $p=1$, the $(\beta,1)$-WWL constraint is the $(\beta-1,\infty)$-RLL constraint. According to Theorem \ref{th:upper}, $C(1,\beta,p)$ is upper bounded by the capacity of the $(\beta,p)$-WWL constraint, $C_{WWL}(\beta,p)$. Thus, we are interested in finding the capacity of this constraint. The approach is similar to the one used in~\cite{WB98} in order to find an upper bound on the capacity of the checkerboard constraint.

\begin{definition}
A {\em merge} of two vectors $\bfu$ and $\bfv$ of the same length $n$ is a function:

$$f_n:\{0,1\}^{n}\times\{0,1\}^{n}\mapsto \{0,1\}^{n+1}\cup{\textbf{\{F\}}}.$$
If the last $n-1$ bits of $\bfu$ are the same as the first $n-1$ bits of $\bfv$, the vector $f_n(\bfu,\bfv)$ is the vector $\bfu$ concatenated with the last bit of $\bfv$, otherwise $f_n(\bfu,\bfv) = \textbf{F}$.
\end{definition}

\begin{definition} \label{defn:transition matrix}
Let $\beta,p$ be two positive integers. Let $S_{\beta,p}$ denote the set of all vectors of length $\beta-1$ having at most $p$ 1's. That is,
$ S_{\beta,p} =  \{\bfs\in\{0,1\}^{\beta-1}:wt(\bfs)\leq p\}$. The size of the set $S_{\beta,p}$ is $M=\sum_{i=0}^{p}\binom{\beta-1}{i}$. Let $\bfs_1,\bfs_2,\ldots,\bfs_M$ be an ordering of the vectors in $S_{\beta,p}$. The {\em transition matrix} for the $(\beta,p)$-WWL constraint, $A_{\beta,p}=(a_{i,j})\in \{0,1\}^{M\times M}$ is defined as follows:

\begin{displaymath}
a_{i,j} = \left\{ \begin{array}{ll}
1 & \textrm{if $f_{\beta-1}(\bfs_i,\bfs_j)\neq \textbf{F}$ and $wt(f_{\beta-1}(\bfs_i,\bfs_j))\leq p$,}\\
0 & \textrm{otherwise.}
\end{array} \right.
\end{displaymath}
\end{definition}

\begin{example}
The following illustrates the construction of the $A_{\beta=3,p=2}$ transition matrix. Note that
\[
S_{3,2}=\{\bfs_1,\bfs_2,\bfs_3,\bfs_4\} = \{(0,0),(0,1),(1,0),(1,1)\},
\]
The merge of $\bfs_i$ and $\bfs_j$ for $i,j=1,2,3,4$ determines the matrix $A_{3,2}$. For example, $f_2(\bfs_1,\bfs_1) = (0,0,0)$, $a_{1,1}=1$; $f_2(\bfs_2,\bfs_1)={\textbf{\em F}}$, $a_{2,1} = 0$; $f_2(\bfs_1,\bfs_2) = (0,0,1)$, $a_{1,2}=1\neq a_{2,1}$. This shows that the matrix is not necessarily symmetric. Finally, $f_2(\bfs_3,\bfs_3)=(1,1,1)$, and $a_{3,3}=0$ since $(1,1,1)$ does not satisfy the (3,2)-WWL constraint.

\[
A_{3,2}=
\left(
    \begin{array}{cccc}
        1 & 1 & 0 & 0\\
        0 & 0 & 1 & 1\\
        1 & 1 & 0 & 0\\
        0 & 0 & 1 & 0
        \end{array}
\right).\\
\]
\end{example}
\begin{definition} \label{defn:irreducible}
A matrix $A\in \{0,1\}^{M\times M}$ is {\em irreducible} if for all $1\leq i,j\leq M$ there exists some $n\geq0$ such that $(A^n)_{i,j}>0$. Note that $n$ can be a function of $i$ and $j$.
\end{definition}
\begin{lemma}\label{lem:irreducible}
For positive integers $\beta,p$, the transition matrix $A_{p,\beta}$ is irreducible.
\end{lemma}
\begin{proof}[Proof of Lemma \ref{lem:irreducible}]
In our construction of $A_{\beta,p}$, it is possible to show that $(A_{\beta,p}^n)_{i,j}$ is the number of vectors of length $n+\beta-1$ starting in $\bfs_i$, ending in $\bfs_j$ and satisfying $(\beta,p)$-WWL constraint, where $\bfs_i$ and $\bfs_j$ are defined in Definition \ref{defn:transition matrix}. Therefore, $A_{\beta,p}$ is irreducible if such a vector of length $n\geq 1$ exists such that it starts with $\bfs_i$ and ends in $\bfs_j$, for every pair of $(i,j)$. Clearly it exists since adding zeros between $\bfs_i$ and $\bfs_j$ gives such a valid vector. This proves the irreducibility of $A_{\beta,p}$.
\end{proof}

The next theorem is a special case of Theorem 3.9 in~\cite{MRS98}.

\begin{theorem} \label{thm:cap-wwl}
The capacity of the $(\beta,p)$-WWL constraint is
\[
C_{WWL}(\beta,p) = \log_2(\lambda_{max}),
\]
where $\lambda_{max}$ is the largest real eigenvalue of $A_{\beta,p}$.
\end{theorem}
\begin{proof}
See Theorem 3.9 in~\cite{MRS98}.
\end{proof}

Fig.~\ref{figure:upper-space} shows $C_{WWL}(\beta,p)$, which is the upper bound of $C(1,\beta,p)$, for $p=1,2,3,4$ respectively.

\begin{figure}[htbp]
    \centering
\includegraphics[width=0.8\linewidth]{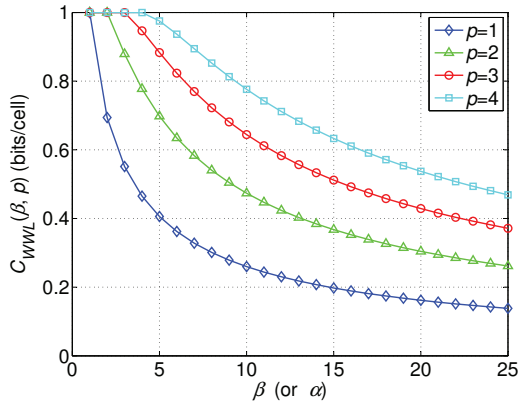}
  \caption{Upper bound on $C(1,\beta,p)$}
  \label{figure:upper-space}
\end{figure}

\begin{remark}
There is a way of presenting the $(\beta,p)$-WWL constraint using labeled graphs (state transition diagrams). We present an example of the labeled graph (transition diagram) for the $(\beta=7,p=2)$-WWL constraint in Fig.~\ref{fig:diag}. An $(\beta,p)$-WWL vector can be generated by reading off the labels along paths in the graph and the sequences in the ellipses indicate the six most recent digits generated.

\begin{figure}[htbp]
  \centering
  \includegraphics[width=0.85\linewidth]{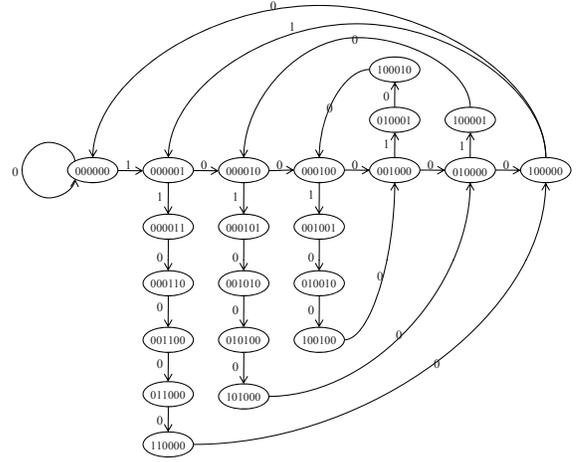}

  \caption{Labeled graphs presenting the (7,2)-WWL constraint}
  \label{fig:diag}
\end{figure}
\end{remark}

\begin{remark}
According to Theorem~\ref{th:upper}, the capacity $C(\alpha,1,p)$ is also upper bounded by the capacity of the $(\alpha,p)$-WWL constraint, $C_{WWL}(\alpha,p)$. Jiang et al.~\cite{JBL10} proposed an upper bound on the rate of an $(\alpha,1,1)$-constrained code with fixed block length $n$ and multiple cell levels. By numerical experiments, we find that their upper bound appears to converge to our upper bound for binary cells when $n \rightarrow \infty$.
\end{remark}

\section{Lower Bound on the Capacity}\label{sec:lower bound}\noindent
In this section, we give lower bounds on the capacity of the $(\alpha,\beta,p)$-constraint based upon specific code constructions. The first construction we give is a trivial one which achieves rate $\frac{p}{\alpha\beta}$. Then, we will show how to improve it for the cases $(1,\beta,p)$ and $(\alpha,1,p)$. In this section we assume that for all positive integers $x,y$ the value of $x (\bmod y)$ belongs to the group $\{1,\ldots,y\}$ via the correspondence $0 \rightarrow y$.

The idea of Construction~\ref{const:trivial} is to partition the set of $n$ cells into subblocks of size $\beta$. Suppose $p = \beta (q-1) + r$, where $1\leq q\leq \alpha$ and $1\leq r \leq \beta$. The encoding process has a period of $\alpha$ writes. On the first $q-1$ writes, all cells in each subblock are programmed with no constraint. On the $q$-th write, the first $r$ cells in each subblock are programmed with no constraint and the rest of the cells are not programmed (staying at level 0). On the $(q+1)$-st to the $\alpha$-th write, no cells are programmed. The details of the construction are as follows.

\begin{const}\label{const:trivial}
Let $\alpha,\beta,p$ be positive integers. We construct an $(\alpha,\beta,p)$-constrained code $\cC$ of length $n$ as follows. To simplify the construction, we assume that $\beta|n$. Let $q=\left\lceil\frac{p}{\beta}\right\rceil, r=p(\bmod \beta).$ For all $i\geq 1$, on the $i$-th write, the encoder uses the following rules:
\begin{itemize}
\item If $1\leq i(\bmod \alpha) < q$, $n$ bits are written to the $n$ cells.
\item If $i(\bmod \alpha) = q$, $rn/\beta$ bits are written in all cells $c_j$ such that $1\leq j(\bmod \beta) \leq r$.
\item If $i(\bmod \alpha) > q$, no information is written to the cells.
\end{itemize}
The decoder is implemented in a very similar way.
\end{const}
\begin{example}
Fig.~\ref{fig:trivial} shows a typical writing sequence of an $(\alpha=3,\beta=3,p=2)$-constrained code of length 15 based on Construction \ref{const:trivial}. The $i$-th row corresponds to the cell-state vector before the $i$-th write. The cells in the box in the $i$-th row are the only cells that can be programmed on the $i$-th write. It can be seen that the rate of the code is the ratio between the number of boxed cells and the total number of cells, which is $\frac{2}{9}$.
\begin{figure}[htbp]
  \centering
  \includegraphics[width=0.7\linewidth]{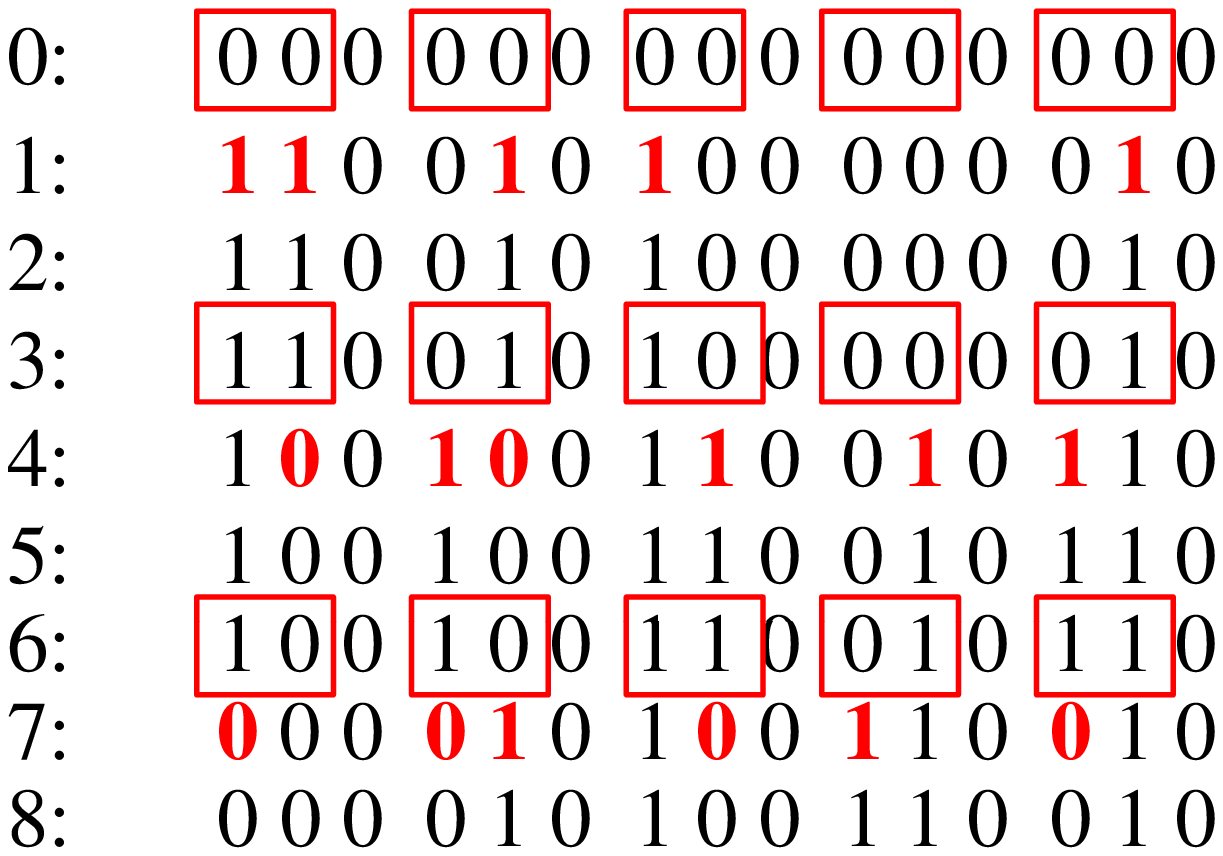}\\
  \caption{A sequence of writes of a $(3,3,2)$-constrained code}
  \label{fig:trivial}
\end{figure}
\end{example}

\begin{theorem} \label{thm:trivial}
The code $\cC$ constructed in Construction~\ref{const:trivial} is an $(\alpha,\beta,p)$-constrained code and its rate is $R=\frac{p}{\alpha\beta}$.
\end{theorem}
\begin{proof}
We show that for all $i\geq 1$ and $1\leq j\leq n-\beta+1$, the rewrite cost of the cells $c_j,c_{j+1},\ldots,c_{j+\beta-1}$ over the writes $i,i+1,\ldots,i+\alpha-1,$ is at most $p$. For all $0\leq k\leq \alpha-1$ such that $1\leq (i+k)(\bmod \alpha) < x$, all of the $\beta$ cells can be written and since there are $x-1$ such values the rewrite cost on these writes is at most $(x-1) \beta$. For $k$, such that $(i+k)(\bmod \alpha) = x$, at most $y$ out of these $\beta$ cells are programmed and therefore the rewrite cost is at most $y$. For all other values of $k$ no other cells are programmed. Therefore, the total rewrite cost is at most $$(x-1)\cdot \beta + y = \left(\left\lceil\frac{p}{\beta}\right\rceil-1\right)\beta + (p\bmod \beta)=p.$$

The total number of bits written on these $\alpha$ writes is $pn/\beta$ and hence the rate of the code is
$$R=\frac{pn/\beta}{\alpha n} = \frac{p}{\alpha\beta}.$$
\end{proof}

\subsection{Space Constraint Improvement}\noindent
In this subsection, we improve the lower bound on $C(1,\beta,p)$ over that offered by the trivial construction. Let $\cS_n(\beta,p)$ be the set of all $(\beta,p)$-WWL vectors of length $n$. We define a \textbf{$(\beta,p)$-WWL code} $\cC_{WWL}$ of length $n$ as a subset of $\cS_n(\beta,p)$. If the size of the code $\cC_{WWL}$ is $M$, then it is specified by an encoding map $\cE_{WWL}:\{1,\ldots,M\}\mapsto\cC_{WWL}$ and a decoding map $\cD_{WWL}:\cC_{WWL}\mapsto\{1,\ldots,M\}$, such that for all $m\in\{1,\ldots,M\}$, $\cD_{WWL}(\cE_{WWL}(m)) = m$.

The problem of finding $(\beta,p)$-WWL codes that achieve the capacity $C_{WWL}(\beta,p)$ is of independent interest and we address it next. Cover~\cite{C73} provided an enumerative scheme to calculate the lexicographic order of any sequence in the constrained system.  For the special case of $p = 1$, corresponding to RLL block codes, Datta and McLaughlin~\cite{DM99,DM01} proposed enumerative methods for binary $(d,k)$-RLL codes based on permutation codes. For $(\beta,p)$-WWL codes, we find enumerative encoding and decoding strategies with linear complexity enumerating all $(\beta,p)$-WWL vectors. We present the coding schemes and the complexity analysis in Appendix~\ref{app:A}. In the sequel, we will simply assume that there exist such codes with rate arbitrarily close to the capacity as the block length goes to infinity for all positive integers $p$ and $\beta$. The next construction uses $(\beta,p)$-WWL codes to construct $(1,\beta,p)$-constrained codes.

\begin{const} \label{const:1 beta p}
Let $\beta,p$ be positive integers such that $p\leq \beta$. Let $\cC_{WWL}$ be a $(\beta,p)$-WWL code of length $n'$ and size $M$. Let $\cE_{WWL}$ and $\cD_{WWL}$ be its encoding and decoding maps. A $(1,\beta,p)$-constrained code $\cC_{1,\beta,p}$ of length $n=2n'+\beta-1$ and its encoding map $\cE$ and decoding map $\cD$ are constructed as follows.
\begin{enumerate}
\item The encoding map $\cE: \{1,\ldots,M\} \times \{0,1\}^n\rightarrow \{0,1\}^n$ is defined for all $(m,\bfu)\in\{1,\ldots,M\} \times \{0,1\}^n$ to be $\cE((m,\bfu))=\bfv$, where
    \begin{enumerate}
    \item $\bfv_1^{n'} = \bfu_1^{n'} + \cE_{WWL}(m)$.
    \item $\bfv_{n'+1}^{n'+\beta-1}=\bf0$,
    \item $\bfv_{n'+\beta}^{n}=\bfu_{1}^{n'}$,
    \end{enumerate}
\item The decoding map $\cD: \{0,1\}^n \rightarrow \{1,\ldots,M\}$ is defined for all $\bfu\in\{0,1\}^n$ to be

$$\cD(\bfu) = \cD_{WWL}(\bfv_1^{n'} + \bfv_{n'+\beta}^{n}).$$
\end{enumerate}
\end{const}
\begin{example}
Here is an example of an $(\alpha=1,\beta=3,p=2)$ code with $n'=4$ for the first 4 writes. The message set has size $M_{n'}=13$ (See the definition of $M_{n'}$ in Definition~\ref{def:order}). The length of the memory is $2n'+\beta-1=10$. Suppose on the second write, the message is $m=7$. Since lexicographically the seventh element in $\cS_4(3,2)$ is $(0110)$, the encoder will copy the previous left block $(1011)$ to the right block and flip the second and the third bits in the left block $(1011)\to(1101)$.
\[
\begin{array}{llcccc|cc|cccc}
& &0 & 0 & 0 & 0    & 0 & 0 &   0 & 0 & 0 & 0  \\
\textrm{1st write},&m=11: & \color{red}{\boldsymbol 1} & 0 & \color{red}{\boldsymbol1} & \color{red}{\boldsymbol1}    & 0 & 0 &   0 & 0 & 0 & 0  \\
\textrm{2nd write},&m=7: & 1 & \color{red}{\boldsymbol1} & \color{red}{\boldsymbol0} & 1                 & 0 & 0 &   \color{red}{\boldsymbol1} & 0 & \color{red}{\boldsymbol1} & \color{red}{\boldsymbol1}  \\
\textrm{3rd write},&m=13: & \color{red}{\boldsymbol0} & \color{red}{\boldsymbol0} & 0 & \color{red}{\boldsymbol0}    & 0 & 0 &   1 & \color{red}{\boldsymbol1} & \color{red}{\boldsymbol0} & 1  \\
\textrm{4th write},&m=4: &0 &  0 & \color{red}{\boldsymbol1} & \color{red}{\boldsymbol1}    & 0 & 0 &   \color{red}{\boldsymbol0} & \color{red}{\boldsymbol0} & 0 & \color{red}{\boldsymbol0}
 \end{array}
\]
\end{example}
\begin{theorem} \label{thm:const-1-beta-p}
The code $\cC_{1,\beta,p}$ is a $(1,\beta,p)$-constrained code. If the rate of the code $\cC_{WWL}$ is $R_{WWL}$, then the rate of the code $\cC_{1,\beta,p}$ is $\frac{n'}{2n'+\beta-1}\cdot R_{WWL}$. Both the encoder and decoder of $\cC_{1,\beta,p}$ have complexity $O(n)$.
\end{theorem}
\begin{proof}
Let $\bfu$ be the cell-state vector in Construction \ref{const:1 beta p}.
\begin{enumerate}
\item
For $\bfu_1^{n'}$, encoder step a) guarantees that the positions of rewritten cells satisfy $(\beta,p)$-WWL constraint. So there are at most $p$ reprogrammed cells in any $\beta$ consecutive cells in $\bfu_1^{n'}$.
\item
For $\bfu_{n'+\beta}^{n}$, three consecutive writes should be examined. Let $\bfw,\bfv,\bfu$ be the cell-state vectors before the $i$-th, $(i+1)$-st, $(i+2)$-nd writes, $i\geq 1$. Encoder step a) means that $\bfv_1^{n'} = \bfw_1^{n'} + \cE_{WWL}(m_i)$, where $m_i \in \{1,\ldots,M\}$ is the message to encode on the $i$-th write. Since encoder step c) guarantees that $\bfv_{n'+\beta}^{n}=\bfw_1^{n'}$ and $\bfu_{n'+\beta}^{n}=\bfv_1^{n'}$, we have $\bfu_{n'+\beta}^{n} = \bfv_{n'+\beta}^{n} + \cE_{WWL}(m_i)$. This proves that $\bfu_{n'+\beta}^{n}$ satisfies the $(1,\beta,p)$ constraint.
\item
For $\bfu_{n'+1}^{n'+\beta-1}$, the cell levels are always set to be 0, which ensures that no violation of the constraint happens between $\bfu_1^{n'}$ and $\bfu_{n'+\beta}^{n}$.
\end{enumerate}
On each write, one of $M$ messages is encoded as a vector of length $n$. Hence, the rate is $\frac{\log_2 M }{n}=\left(\frac{\log_2 M }{n'} \frac{n'}{2n'+\beta-1}\right)=\frac{n'}{2n'+\beta-1}\cdot R_{WWL}$.

The encoder $\cE$ and decoder $\cD$ come directly from $\cE_{WWL}$ and $\cD_{WWL}$, which have complexity $O(n)$ both in time and in space. Therefore, $\cE$ and $\cD$ both have linear complexity in time and in space.
\end{proof}

\begin{corollary}\label{cor:lower 1 beta p}
Let $\beta,p$ be two positive integers such that $p\leq \beta$, then
$$C(1,\beta,p) \geq \max \left\{\frac{C_{WWL}(\beta,p)}{2},\frac{p}{\beta} \right\}.$$
\end{corollary}

Corollary~\ref{cor:lower 1 beta p} provides a lower bound that is achieved by practical coding schemes. In fact, following similar proofs in~\cite{Berger98,Cohen83,CZ88}, we can prove the following theorem using probabilistic combinatorial tools~\cite{AS92}.
\begin{theorem} \label{thm:cap 1 beta p}
Let $\beta,p$ be positive integers such that $\beta\geq p$, then
$$C(1,\beta,p) = C_{WWL}(\beta,p),$$
where $C_{WWL}(\beta,p)$ is the capacity of the $(\beta,p)$-WWL constraint.
\end{theorem}
\begin{proof}
See Appendix~\ref{app:B}.
\end{proof}

\subsection{Time Constraint Improvement}\noindent
Jiang et al. constructed in~\cite{JBL10} an $(\alpha,1,1)$-constrained code. Let us explain their construction as it serves as the basis for our construction. Their construction uses Write-Once Memory (WOM)-codes~\cite{RS82}. A WOM is a storage device consisting of cells that can be used to store any of $q$ values. In the binary case, each cell can be irreversibly changed from state 0 to state 1. We denote by $[n,t;2^{nR_1},\ldots,2^{nR_t}]$ a $t$-write WOM code $\cC_W$ such that the number of messages that can be written to the memory on its $i$-th write is $2^{nR_i}$, and the sum-rate of the WOM code is defined to be $R_{\textrm{sum}}=\sum_{i=1}^tR_i$. The sum-capacity $C_{\textrm{sum}}$ is defined as the supremum of achievable sum-rates. The code is specified by $t$ pairs of encoding and decoding maps, $(\cE_i,\cD_i)$, where $i\in\{1,2,\ldots,t\}$. Assuming that the cell-state vector before the $i$-th write is $\bfc_{i}$, the encoder is a map
\[
\cE_i: [1:2^{nR_i}]\times \{0,1\}^n \rightarrow \{0,1\}^n,
\]
such that for all $(m,\bfc_{i-1,1}^n) \in[1:2^{nR_i}]\times \{0,1\}^n$,
\[
\bfc_{i-1,1}^n \preceq \bfc_{i,1}^n = \cE_i(m,\bfc_{i-1,1}^n),
\]
where the relation ``$\preceq$'' is defined in Definition~\ref{def:order}.
The decoder
\[
\cD_i: \{0,1\}^n \rightarrow [1:2^{nR_i}],
\]
satisfies
\[
\cD_i(\cE_i(m,\bfc_{i-1,1}^n)) = m.
\]
for all $m\in[1:2^{nR_i}]$,

It has been shown in~\cite{H85} that the sum-capacity of a $t$-write WOM is $C_{\textrm{sum}}=\log_2(t+1)$.

The constructed $(\alpha,1,1)$-constrained code has a period of $2(t+\alpha)$ writes. On the first $t$ writes of each period, the encoder simply writes the information using the encoding maps of the $t$-write WOM code. Then, on the $(t+1)$-st write no information is written but all the cells are increased to level one. On the following $\alpha-1$ writes no information is written and the cells do not change their levels; that completes half of the period. On the next $t$ writes the same WOM code is again used; however since now all the cells are in level one, the complement of the cell-state vector is written to the memory on each write. On the next write no information is written and the cells are reduced to level zero. In the last $\alpha-1$ writes no information is written and the cells do not change their values. We present this construction now in detail.

\begin{const}\label{const:alpha 1 1}
Let $\alpha$ be a positive integer and let $\cC_W$ be an $[n,t;2^{nR_1},\ldots,2^{nR_t}]$ $t$-write WOM code. Let $\cE_i(m,\bfv_{i-1})$ be the $i$-th encoder of $\cC_W$, for $m\in[1:2^{nR_i}], i\in[1:t]$. An $(\alpha,1,1)$-constrained code $\cC_{\alpha,1,1}$ is constructed as follows. For all $i\geq 1$, let $i'=i(\bmod (2(t+\alpha)))$, where $1\leq i'\leq 2(t+\alpha)$. The cell-state vector after the $i$-th write is denoted by $\bfc_i$. On the $i$-th write, the encoder uses the following rules:
\begin{itemize}
\item If $i'\in[1:t]$, write $M_{i'}\in[1:2^{nR_i'}]$ such that
\[
\bfc_i = \cE_{i'}(M_{i'},\bfc_{i-1}).
\]
\item If $i' = t+1$, no information is written and the cell-state vector is changed to the all-one vector \,$\boldsymbol{1}$, i.e., $\bfc_i=\boldsymbol{1}$.
\item If $i'\in[t+2:t+\alpha]$, no information is written and the cell-state vector is not changed.
\item If $i'\in[t+\alpha+1:2t+\alpha]$, write $M_{i'-t-\alpha}\in[1:2^{nR_{i'-t-\alpha}}]$ such that
\[
\bfc_i = \overline{\cE_{i'-t-\alpha}(M_{i'-t-\alpha},\overline{\bfc}_{i-1})}.
\]
\item If $ i' = 2t+\alpha+1$, no information is written and the cell-state vector is changed to the all-zero vector \,$\boldsymbol{0}$, i.e., $\bfc_i=\boldsymbol{0}$.
\item If $ i'\in[2t+\alpha+1:2(t+\alpha)]$, no information is written and the cell-state vector is not changed.
\end{itemize}
\end{const}

\begin{remark}
This construction is presented differently in~\cite{JBL10}. This results from the constraint of having the same rate on each write which we can bypass in this work. Consequently, in our case we can have varying rates and thus the code $\cC_{\alpha,1,p}$ can achieve a higher rate.
\end{remark}

\begin{theorem} \label{thm:const-alpha-1-1}
The code $\cC_{\alpha,1,1}$ is an $(\alpha,1,1)$-constrained code. If the $t$-write WOM code $\cC_W$ is sum-rate optimal, then the rate of $\cC_{\alpha,1,1}$ is $\frac{\log_2(t+1)}{t+\alpha}$.
\end{theorem}
\begin{proof}
In every period of $2(t+\alpha)$ writes, every cell is programmed at most twice; once in the first $t+1$ writes and once in the first $t+1$ writes of the second part of the write-period. After every sequence $t+1$ writes, the cell is not programmed for $\alpha-1$ writes. Therefore the rewrite cost of every cell among $\alpha$ consecutive rewrites is at most $1$.

If the rate of the WOM code $\cC_W$ is $R_W$ then $2nR_W$ bits are written in every period of $2(t+\alpha)$ writes. Hence, the rate of $\cC_{\alpha,1,1}$ is $\frac{2nR_{\small W}}{2(t+\alpha)n} = \frac{R_{\small W}}{t+\alpha}$. If $\cC_W$ is sum-rate optimal, the rate of $\cC_{\alpha,1,1}$ is therefore $\frac{\log_2(t+1)}{t+\alpha}$.
\end{proof}

The next table shows the highest rates of $(\alpha,1,1)$-constrained codes based on Construction \ref{const:alpha 1 1} for $\alpha=4,\ldots,8$.

\begin{center}
\begin{tabular}{|c||c|c|c|c|c|}
  \hline
  $\alpha$ & 4 & 5 & 6 & 7 & 8 \\ \hline
  1/$\alpha$ & 0.25 & 0.2 & 0.167 & 0.143 & 0.125 \\ \hline
  rate of $\cC_{\alpha,1,1}$ & 0.290 & 0.256 & 0.235 & 0.216& 0.201 \\
  \hline
\end{tabular}
\end{center}

Next, we would like to extend Construction \ref{const:alpha 1 1} in order to construct $(\alpha,1,p)$-constrained codes for all $p\geq 2$. For simplicity of the construction, we will assume that $p$ is an even integer; and the required modification for odd values of $p$ will be immediately clear. We choose $t\geq 1$ such that $\alpha \geq (p-1)t$ and the period of the code is $\alpha+t$. On the first $t$ writes of each period, the encoder uses the encoding map of the $t$-write WOM code. In the following $t$ writes, it uses the bit-wise complement of a WOM code as in Construction~\ref{const:alpha 1 1}. This procedure is repeated for $\frac{p}{2}$ times; this completes the first $tp$ writes in the period. On the $(tp+1)$-st write, no new information is written and the cell-state vector is changed to the all-zero vector. During the $(tp+2)$-nd to $(\alpha+t)$-th writes, no information is written and the cell-state vector is not changed. That completes one period of $\alpha+t$ writes.

\begin{const}\label{const:alpha_1_p_1}
Let $\alpha,p,t$ be positive integers such that $\alpha\geq (p-1)t$. Let $\cC_W$ be an $[n,t;2^{nR_1},\ldots,2^{nR_t}]$ $t$-write WOM code. For $i\in[1:t]$, let $\cE_i(m,\bfv_{i-1})$ be its encoding map on the $i$-th write, where $m\in[1:2^{nR_i}]$. An $(\alpha,1,p)$-constrained code $\cC_{\alpha,1,p}$ is constructed as follows. For all $i\geq 1$, let $i'=i(\bmod (\alpha+t))$, $i''=i'(\bmod 2t)$ where $1\leq i'\leq (\alpha+t), 1\leq i''\leq 2t$. The cell-state vector after the $i$-th write is denoted by $\bfc_i$. On the $i$-th write, the encoder uses the following rules:
\begin{itemize}
\item If $i'\in[1:pt]$ and $i''\in[1:t]$, write $M_{i''}\in[1:2^{nR_{i''}}]$ such that
\[
 \bfc_i = \cE_{i''}(M_{i''},\bfc_{i-1}).
\]
\item If $i'\in[1:pt]$ and $i''\in[t+1:2t]$, write $M_{i''-t}\in[1:2^{nR_{i''-t}}]$ such that
\[
\bfc_i = \overline{\cE_{i''-t}(M_{i''-t},\overline{\bfc}_{i-1})}.
\]
\item If $i'=pt+1$, no information is written and the cell-state vector is changed to $\boldsymbol{0}$, i.e., $\bfc_i=\boldsymbol{0}$.
\item If $i'\in[pt+2:\alpha+t]$, no information is written and the cell-state vector is not changed.
\end{itemize}
\end{const}

\begin{theorem} \label{thm:const-alpha-1-p}
The code $\cC_{\alpha,1,p}$ is an $(\alpha,1,p)$-constrained code. If the $t$-write WOM code $\cC_W$ is sum-rate optimal, then the rate of $\cC_{\alpha,1,p}$ is $\frac{p\log_2(t+1)}{\alpha+t}$.
\end{theorem}
\begin{proof}
This is similar to the proof of Theorem~\ref{thm:const-alpha-1-1}, so we present here only a sketch of the proof. In every period of $(\alpha+t)$ writes, each cell is rewritten at most $p$ times. In particular, the first rewrite happens before the $(t+1)$-st write. After that, the cell is rewritten at most $p-1$ times until the $(tp+1)$-st write and then not programmed for $\alpha+t-(tp+1)$ writes. Therefore, each cell is rewritten at most $p$ times on $\alpha+t-(tp+1)+(tp+1)-t =\alpha$ writes. This proves the validity of the code.

If the rate of the WOM code $\cC_W$ is $R_W$ then $pn R_W$ bits are written during each period of $\alpha+t$ writes since the WOM code is used $p$ times. Hence, the rate of $\cC_{\alpha,1,p}$ is $\frac{2pn R_W}{2(\alpha+t)n}=\frac{pR_W}{\alpha+t}$. If that $\cC_W$ is sum-rate optimal, the rate of $\cC_{\alpha,1,p}$ is $\frac{p\log_2(t+1)}{\alpha+t}$.
\end{proof}

\begin{remark} \label{rmk:lower alpha 1 p}
In Construction~\ref{const:alpha_1_p_1} we required that $\alpha\geq (p-1)t$ and, in particular, $t\leq \left\lfloor\frac{\alpha}{p-1}\right\rfloor$. If $t>\left\lfloor\frac{\alpha}{p-1}\right\rfloor$, we can simply use Construction~\ref{const:alpha_1_p_1} while taking $\alpha = (p-1)t$, i.e., the period of writes is now $pt$ and and we construct a $((p-1)t,1,p)$-constrained code, which is also an $(\alpha,1,p)$-constrained code. The rate of the code is $R_W/t$, where $R_W$ is the rate of the WOM code $\cC_W$.
\end{remark}

The next corollary provides lower bounds on $C(\alpha,1,p)$.

\begin{corollary} \label{thm:lower alpha-1-p}
Let $\alpha,p$ be positive integer such that $p \leq \alpha$. Then,
\[
C(\alpha,1,p)\geq \max_{t, t^*\in \mathbb{Z}_+,} \left\{\frac{p\log_2(t+1)}{\alpha+t}, \frac{\log_2(t^*+1)}{t^*},\frac{p}{\alpha}\right\},
\]
where
\[
1\leq t \leq \left\lfloor\frac{\alpha}{p-1}\right\rfloor, t^*=\left\lceil\frac{\alpha}{p-1}\right\rceil.
\]

\end{corollary}

\begin{figure}[htbp]
    \centering
  \includegraphics[width=0.8\linewidth]{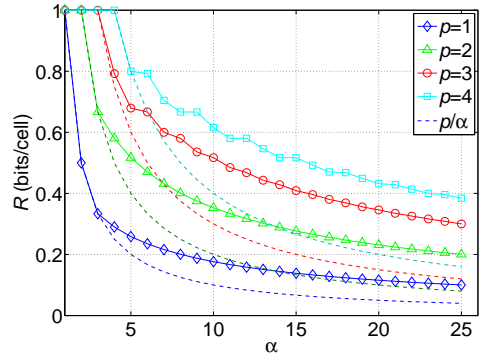}
  \caption{Lower bound on $C(\alpha,1,p)$}
\label{figure:lower case time}
\end{figure}
Figure \ref{figure:lower case time} shows the rates of $(\alpha,\beta=1,p)$ constrained codes obtained by selecting the best $t$ for each pair of $(\alpha,p)$. In comparison to the codes in Construction \ref{const:trivial} whose rates are shown by the dashed lines, our construction approximately doubles the rates. Our lower bounds achieve approximately $78\%$ of the corresponding upper bounds on $C(\alpha,1,p)$.

\subsection{Time-Space Constraint Improvement}\noindent

In this section, we are interested in combining the improvements in time and in space to provide lower bounds on the capacity of $(\alpha,\beta,p)$-constraints.

\begin{theorem}\label{thm:const-alpha-beta-p}
For all $\alpha, \beta,p$ positive integers,
\[
C(\alpha,\beta,p)\geq \max\left\{\frac{C(\alpha,1,p)}{\beta},\frac{C(1,\beta,p)}{\alpha}\right\}.
\]
\end{theorem}
\begin{proof}
An $(\alpha,\beta,p)$-constrained code can be constructed in two ways.
\begin{enumerate}
\item
Let $\cC$ be a $(1,\beta,p)$-constrained code of rate $R$ and length $n$. We construct a new code $\cC'$ with the same number of cells. New information is written to the memory on all $i$-th writes, where $i\equiv 1(\bmod \alpha)$, simply by using the $\left\lceil\frac{i}{\alpha}\right\rceil$-th write of the code $\cC$. Then, the code $\cC'$ is an $(\alpha,\beta,p)$-constrained code and its rate is $R/\alpha$. Therefore, we conclude that $C(\alpha,\beta,p)\geq \frac{C(1,\beta,p)}{\alpha}$.
\item
Let $\cC$ be an $(\alpha,1,p)$-constrained code of rate $R$ and length $n$. We construct a new code $\cC'$ for $n\beta$ cells: $(c_1,c_2,\ldots,c_{n\beta})$. The code $\cC'$ uses the same encoding and decoding maps of the code $\cC$, while using only the $n$ cells $c_i$ such that $i\equiv1(\bmod \beta)$. Then, the code $\cC'$ is an $(\alpha,\beta,p)$-constrained code and its rate is $R/\beta$. Therefore, we conclude that $C(\alpha,\beta,p)\geq \frac{C(\alpha,1,p)}{\beta}$.
\end{enumerate}
The capacity must be greater than or equal to the maximum of the two lower bounds.
\end{proof}


\appendices
\section{}\label{app:A}
In this section, we show an enumerative encoding and decoding strategy with linear complexity for the set of $(\beta,p)$-WWL vectors.

\begin{definition}\label{def:order}
Let $\bfX=\{\bfx_1,\ldots,\bfx_N\}$ be a set of distinct binary vectors, $\bfx_i\in\{0,1\}^n, i=1,\ldots,N$. Let $\psi(\bfx)$ denote the decimal representation of a vector $\bfx\in\{0,1\}^n$. For $\bfx,\bfy\in\{0,1\}^n$, we say $\bfx\preceq\bfy$ (or $\bfx\prec\bfy$) if and only if $\psi(\bfx)\leq\psi(\bfy)$ (or $\psi(\bfx)<\psi(\bfy)$). The {\em order} of the element $\bfx_i$ in $\bfX$ is defined as:
\[
ord(\bfx_i) = \big| \{ j: \bfx_j \preceq \bfx_i, 1\leq j\leq N\} \big|.
\]
\end{definition}

Let $\{\bfc_1,\ldots,\bfc_{M_n}\}$ be an ordering of the elements in $\cS_n(\beta,p)$, where $M_n=|\cS_n(\beta,p)|$. The encoder and decoder of a $(\beta,p)$-WWL code give a one-to-one mapping between $\cS_n(\beta,p)$ and $\{1,\ldots,M_n\}$, namely $\cE_{WWL}(m)=\bfc_m$ where $o(\bfc_m)=m$ and $\cD_{WWL}(\bfc_m)=o(\bfc_m)=m$, for all $m = \{1,\ldots,M_n\}$. Now the problem is to calculate $o(\bfc_m)$ given $\bfc_m$. Let $\bfs_1,\ldots,\bfs_{M_{\beta-1}}$ be the ordering of the vectors in $S_{\beta,p}$ introduced in Definition~\ref{defn:transition matrix}, where $M_{\beta-1}=|S_{\beta,p}|=|\cS_{\beta-1}(\beta,p)|=\sum_{i=0}^{p}\binom{\beta-1}{i}$. Let
$$\bfx_{\beta,p,n}=(x_1(n),x_2(n),\ldots,x_{M_{\beta-1}}(n))^T,$$
where $x_i(n)$ is the number of $(\beta,p)$-WWL vectors of length $n$ that have the vector $\bfs_i$ as a prefix, where $\bfx^T$ denotes the transpose of $\bfx$.
\begin{lemma} \label{lem:recursion}
The vectors $\bfx_{\beta,p,n+1}$, $n\geq\beta$, satisfy the first-order recursion:
\[ \bfx_{\beta,p,n+1} = A_{\beta,p}\cdot \bfx_{\beta,p,n}. \]
\end{lemma}
\begin{proof}
See~\cite{WB98}.
\end{proof}

The encoder and decoder have access to a matrix $\BFX_{\beta,p,n} \in \bfZ_+^{(n+\beta)\times M_{\beta-1}}$, where the $i$-th row of $\BFX_{\beta,p,n}$ is $\bfx_{\beta,p,i}^T$, $i=1,\ldots,n+\beta$. For simplicity, $\BFX_{\beta,p,n}$ is written as $\BFX$ if no confusion can occur. We denote by $\BFX(i,j)$ the entry in the $i$-th row and $j$-th column of $\BFX$ and we define $\BFX(i,:)$, $\BFX(:,j)$

\noindent  to be the $i$-th row vector, $j$-th column vector of $\BFX$, respectively, i.e., $\BFX(i,:)=(\BFX(i,1),\ldots,\BFX(i,M_{\beta-1}))$ and $\BFX(:,j) = (\BFX(1,j),\ldots,\BFX(n+\beta,j))^T$. From Lemma~\ref{lem:recursion}, $\BFX_{\beta,p,n}$ can be calculated efficiently with time complexity $O(n)$.


%
%
%
%
%
%
%

\subsubsection{Decoder}
Based on $\BFX_{\beta,p,n}$, we present an enumerative method to calculate the order of each element in $\cS_n(\beta,p)$. Note that the order of a vector is the decoded message corresponding to that vector. In this algorithm, the decoder scans the vector from left to right. Whenever the decoder finds a 1 in the vector, the order of the vector will increase. The details of the algorithm are presented below. Here $\bfc=(c_1,\ldots,c_n)\in\cS_n(\beta,p)$ is the binary vector to be decoded; the algorithm calculates $o(\bfc)\in\{1,\ldots,M_n\}$.
\begin{algr}  \label{algr:decode}

\textsc{Decoding: Calculate} $o(\bfc)$, $\bfc\in\cS_n(\beta,p)$

{\em1:}\noindent\hspace*{\paragraphindent}
 let $cnt = 0, j = 1, i = 0$;\vspace{1mm}

{\em2:}\noindent\hspace*{\paragraphindent}
 while $(i\leq n)$\{\vspace{1mm}

{\em3:}\noindent\hspace*{2\paragraphindent}
 while $(j\leq n \text{ and } c(j)\neq1)$\vspace{1mm}

{\em4:}\noindent\hspace*{3\paragraphindent}
 $j = j + 1$; \vspace{1mm}

{\em5:}\noindent\hspace*{2\paragraphindent}
if $(j = n + 1)$ \vspace{1mm}

{\em6:}\noindent\hspace*{3\paragraphindent}
$o(\bfc) = cnt + 1$;\vspace{1mm}

{\em7:}\noindent\hspace*{3\paragraphindent}
algorithm ends;\vspace{1mm}

{\em8:}\noindent\hspace*{2\paragraphindent}
\}

{\em9:}\noindent\hspace*{2\paragraphindent}
\textcolor{ForestGreen}{/*A $1$ is detected in $\bfc$.*/}\vspace{1mm}

{\em10:}\noindent\hspace*{2\paragraphindent}
let $\bfd = (0,\ldots,0)$ with length $\beta-1$;\vspace{1mm}

{\em11:}\noindent\hspace*{0\paragraphindent}
\textcolor{ForestGreen}{/*$\bfd$ is a vector storing $\beta-2$ bits to the left of the detected 1, appended with a 0.*/}\vspace{1mm}

{\em12:}\noindent\hspace*{2\paragraphindent}
if $(j\geq \beta-1)$ \vspace{1mm}

{\em13:}\noindent\hspace*{3\paragraphindent}
let $\bfd_1^{\beta-2} = \bfc_{j-\beta+2}^{j-1}$;\vspace{1mm}

{\em14:}\noindent\hspace*{2\paragraphindent}
else \textcolor{ForestGreen}{/*$j<\beta-1$*/} \vspace{1mm}

{\em15:}\noindent\hspace*{3\paragraphindent}
let $\bfd_{\beta-j}^{\beta-2} = \bfc_1^{j-1}$;\vspace{1mm}

{\em16:}\noindent\hspace*{2\paragraphindent}
find $k\in[1:M_{\beta-1}]$ such that $\bfs_k= \bfd $;
\vspace{1mm}


{\em17:}\noindent\hspace*{2\paragraphindent}
$cnt = cnt + \BFX(n-j+\beta-1,k)$;\vspace{1mm}

{\em18:}\noindent\hspace*{2\paragraphindent}
$i = j$; $j = i+1$;\vspace{1mm}

{\em19:}\noindent\hspace*{1\paragraphindent}
\}\vspace{1mm}

{\em20:}\noindent\hspace*{\paragraphindent}
$o(\bfc) = cnt + 1$; \vspace{1mm}

{\em21:}\noindent\hspace*{\paragraphindent}
algorithm ends.
\end{algr}
\begin{example}
Suppose we would like to decode a $(6,3)$-WWL vector $\bfc=(1 0 1 1 0 0 1 0 0 1)$ of length 10.
\begin{itemize}
\item A 1 is detected $({\color{red}{\boldsymbol 1}} 0 1 1 0 0 1 0 0 1)$, where $i = 0,j=1$. The decoder aims to find the number of vectors $\hat\bfc$ such that $(0 0 0 0 0 0 0 0 0 0) \preceq \hat{\bfc} \prec (1 0 0 0 0 0 0 0 0 0)$. Now $\bfd=(00000)=\bfs_1$, so $k=1$, and $n-j+\beta-1 = 14$. Therefore, $cnt = 0 + \BFX_{6,3,16}(14,1)=236$.

\item A 1 is detected $(1 0 {\color{red}{\boldsymbol 1}} 1 0 0 1 0 0 1)$, where $i = 1,j = 3$. The decoder aims to find the number of vectors $\hat\bfc$ such that $(1 0 0 0 0 0 0 0 0 0) \preceq \hat\bfc \prec (1 0 1 0 0 0 0 0 0 0)$. Here $\bfd=(00100)=s_5$, so $k=5$, and $n-j+\beta-1 = 12$. Therefore, $cnt = 236 + \BFX_{6,3,16}(12,5)=308$.


\item  A 1 is detected $(1 0  1 {\color{red}{\boldsymbol 1}} 0 0 1 0 0 1)$, where $i = 3,j = 4$. The decoder aims to find the number of vectors $\hat\bfc$ such that $(1 0 0 1 0 0 0 0 0 0) \preceq \hat\bfc \prec (1 0 1 1 0 0 0 0 0 0)$. Here $\bfd=(01010)=s_{11}$, so $k=11$, and $n-j+\beta-1 = 11$. Therefore, $cnt = 308 + \BFX_{6,3,16}(11,11)=343$.

\item A 1 is detected $(1 0  1 1 0 0 {\color{red}{\boldsymbol 1}} 0 0 1)$, where $i = 4,j = 7$. The decoder aims to find the number of vectors $\hat\bfc$ such that $(1 0 1 1 0 0 0 0 0 0) \preceq \hat\bfc \prec (1 0 1 1 0 0 1 0 0 0)$. Here $\bfd=(11000)=s_{23}$, so $k=23$, and $n-j+\beta-1 = 8$. Therefore, $cnt = 343 + \BFX_{6,3,16}(8,23)=351$.
\item Finally, a 1 is detected $(1 0  1 1 0 0 1 0 0 {\color{red}{\boldsymbol 1}})$, where $i = 7,j = 10$. The decoder aims to find the number of vectors $\hat\bfc$ such that $(1 0 1 1 0 0 1 0 0 0) \preceq \hat\bfc \prec (1 0 1 1 0 0 1 0 0 1)$. Here $\bfd=(01000)=s_{9}$, so $k=9$, and $n-j+\beta-1 = 5$. Therefore, $cnt = 351 + \BFX_{6,3,16}(5,9)=352$.
\end{itemize}
We calculate that $o(\bfc)= cnt + 1=353$ and $\bfc$ is decoded as 353.
\end{example}

\begin{theorem} \label{thm: complexity decoder}
Algorithm~\ref{algr:decode} calculates the order of a $(\beta,p)$-WWL vector of length $n$ in $\cS_n(\beta,p)$. Its time complexity and space complexity are both $O(n)$.
\end{theorem}
\begin{proof}

{\em Correctness}: Let $\bfc$ be the vector to decode; that is, we seek to find $ord(\bfc)$. For $\bfc_1\preceq\bfc_2$, we denote by $N(\bfc_1,\bfc_2)$ the number of vectors $\hat\bfc$ such that $\bfc_1\preceq\hat\bfc\prec\bfc_2$. Let $\bfc_1,\ldots,\bfc_L$ be a sequence of vectors such that $\textbf{0}=\bfc_0\preceq\bfc_1\preceq\bfc_2\preceq\cdots\preceq\bfc_L=\bfc$; then it is easy to see
\[
o(\bfc) = \sum_{i=1}^L N(\bfc_{i-1},\bfc_i)+1.
\]
Let $L$ be the number of 1's in $\bfc$; let all the indices of 1's be $j_1,j_2,\ldots,j_L$ in ascending order, that is $1\leq j_1<\cdots<j_L\leq n$ and $c_{j_1}=c_{j_2}=\cdots=c_{j_L}=1$. For $i\in\{1,\ldots,L\}$, $\bfc_i$ is chosen such that $\bfc_i = \bfc_{i-1} + \boldsymbol\delta_{j_i}$, where $\bfc_0 = \boldsymbol{0}$, and $\boldsymbol\delta_j$, $j\in\{1,\ldots,n\}$, denotes the vector where all entries are 0 except for the $j$-th entry, which is a 1. Here addition is component-wise modulo-2 summation.

Lines 3 and 4 together with Line 18 in Algorithm~\ref{algr:decode} scan $\bfc$ and find $\bfc_i$ according to $\bfc_{i-1}$. Therefore, we are left to prove that Algorithm~\ref{algr:decode} calculates $N(\bfc_{i-1},\bfc_i)$ for $i\in\{1,\ldots,L\}$.

By definition, the first $j_i-1$ digits of $\bfc_i$ and $\bfc_{i-1}$ are the same, and $c_{i,j_i}=1$ while $c_{i-1,j_i}=0$. Then a vector $\hat\bfc\in\{0,1\}^n$ satisfies $\bfc_{i-1}\preceq\hat\bfc\prec\bfc_i$ if and only if the first $j_i$ digits of $\hat\bfc$ are the same as those of $\bfc_{i-1}$, i.e. $\hat\bfc_{1}^{j_i} = \bfc_{i-1,1}^{j_i}$. Given the length and the first $j_i$ digits of $\hat\bfc$, the number of possible $\hat\bfc$ can be calculated based on the matrix $\BFX$ in the following way. Since the $(\beta,p)$-WWL constraint is local, if $j_i>\beta-1$, the task is equivalent to calculating the number of $\tilde\bfc$ with length $n-j_i+\beta-1$ such that the first $\beta-1$ digits are a prefix of $\hat\bfc$, in particular, $\tilde\bfc_{1}^{\beta-1} = \hat\bfc_{j_i-\beta+2}^{j_i}$; otherwise, for $j_i \leq \beta-1$, it is equivalent to calculating the number of $\tilde\bfc$ with length $n-j_i+\beta-1$ such that the first $\beta-1$ digits are zeros followed by length-$j_i$ prefix of $\hat\bfc$, that is, $\tilde\bfc_1^{\beta-1}=(\textbf{0}_{\beta-1-j_i},\hat\bfc_{1}^{j_i})$. Lines 10 -- 15 in Algorithm~\ref{algr:decode} find the first $\beta-1$ digits of $\tilde\bfc$ and Lines 16 and 17 calculate the number of $\tilde\bfc$, which is the number of vectors $\hat\bfc$ satisfying $\bfc_{i-1}\preceq\hat\bfc\prec\bfc_i$. Therefore, Algorithm~\ref{algr:decode} calculates $N(\bfc_{i-1},\bfc_i)$ for $i\in\{1,\ldots,L\}$ and sums them up to derive the order of $\bfc$.

{\em Time complexity analysis}: It can be seen from the algorithm that the decoder scans the vector that is to be decoded only once. Whenever the decoder detects a 1, it uses binary searches to find the corresponding prefix vector $\bfd$ in $\BFX$, while the number of 1's is no more than $\frac{np}{\beta}$. Therefore, the time complexity of the decoder is no more that $O(\frac{np}{\beta}\log M_{\beta-1})=O(\frac{np}{\beta}\log\sum_{i=0}^{p}\binom{\beta-1}{i})=O(n)$, where $\beta$ and $p$ are fixed integers and not related to $n$.

{\em Space complexity analysis}: The space complexity comes from the matrix $\BFX$ with $n+\beta-1$ rows and $M_{\beta-1}$ columns. Therefore, the space complexity is also $O(n)$ since $\beta$ and $M_{\beta-1}$ are both fixed integers.
\end{proof}

\subsubsection{Encoder}
The encoder follows a similar approach to map an integer $m\in\{1,\ldots,M_n\}$ to a vector $\bfc\in \cS_n(\beta,p)$, such that $o(\bfc)=m$. We call $\bfc$ the {\em encoded vector} for the message $m$.
Note that $\forall m_i,m_j\in\{1,\ldots,M_n\}, m_i\leq m_j$ if and only if $\bfc_i \preceq \bfc_j$, where $o(\bfc_i)=m_i$ and $o(\bfc_j)=m_j$.
The following encoding algorithm uses the matrix $\BFX$ to efficiently calculate the vector $\bfc\in\cS_n(\beta,p)$ such that $o(\bfc)=m$, for $m\in\{1,\ldots,M_n\}$. The algorithm has linear complexity.

\begin{algr} \label{algr:encode}
\textsc{Encoding: find $\bfc$ such that $o(\bfc)=m$}

\noindent\hspace*{\paragraphindent}
let $cnt = 0, \bfc = (0,\ldots,0)$ with length $n$;\vspace{1mm}

\noindent\hspace*{\paragraphindent}
for $i = 1,2,\ldots,n$ \{ \vspace{1mm}

\noindent\hspace*{2\paragraphindent}
let $\bft = \bfc$;\vspace{1mm}

\noindent\hspace*{2\paragraphindent}
let $t(i) = 1$;\vspace{1mm}

\noindent\hspace*{2\paragraphindent}
if $\bft$ satisfies $(\beta,p)$-WWL constraint \{\vspace{1mm}

\noindent\hspace*{3\paragraphindent}
let $\bfq = (0,\ldots,0)$ with length $\beta-1$;\vspace{1mm}

\noindent\hspace*{0\paragraphindent}
\textcolor{ForestGreen}{/*$\bfq$ is a vector storing $\beta-2$ bits to the left of $t(i)$ in $\bft$, appended with a 0.*/}\vspace{1mm}

\noindent\hspace*{3\paragraphindent}
if $(i\geq \beta-1)$ \vspace{1mm}

\noindent\hspace*{4\paragraphindent}
let $\bfq_1^{\beta-2}=\bft_{i-\beta+2}^{i-1}$;\vspace{1mm}

\noindent\hspace*{3\paragraphindent}
else \textcolor{ForestGreen}{/*$i<\beta-1$*/} \vspace{1mm}

\noindent\hspace*{4\paragraphindent}
let $\bfq_{\beta-i}^{\beta-2}=\bft_1^{i-1}$;\vspace{1mm}

\noindent\hspace*{3\paragraphindent}
find $k\in[1:M_{\beta-1}]$ such that $\bfs_k = \bfq$.\vspace{1mm}
%

\noindent\hspace*{3\paragraphindent}
let $CntTry = cnt + \BFX(n-i+\beta-1,k)$;

\noindent\hspace*{3\paragraphindent}
if $(CntTry + 1 = m)$ \{

\noindent\hspace*{4\paragraphindent}
$\bfc = \bft$;

\noindent\hspace*{4\paragraphindent}
return $\bfc$; algorithm ends;

\noindent\hspace*{3\paragraphindent}
\}

\noindent\hspace*{3\paragraphindent}
if $(CntTry + 1 < m)$ \{

\noindent\hspace*{4\paragraphindent}
let $\bfc(i) = 1$;

\noindent\hspace*{4\paragraphindent}
let $cnt = CntTry$;

\noindent\hspace*{3\paragraphindent}
\}

\noindent\hspace*{2\paragraphindent}
\}

\noindent\hspace*{\paragraphindent}
\}
\end{algr}

\begin{example}
Suppose we would like to encode one of $M_n = 421$ $(\beta=6,p=3)$-WWL vectors of length $n=10$. The message to be encoded is $m=353$.
\begin{itemize}
\item
$\bfc = (0000000000)$,  $i = 1$, $\bft = (1000000000)$, $\bfq = (00000)=s_1$, so $k =1$. Since $cnt = 0$, $CntTry = cnt + \BFX(n-i+\beta-1,k)=236<m-1$, so set $cnt = 236$.
\item
$\bfc = (1000000000)$, $i = 2$, $\bft = (1100000000)$, $\bfq = (00010)=s_3$, so $k=3$. Compute $CntTry = cnt + \BFX(n-i+\beta-1,k)=236+\BFX(13,3)=355>m-1$.
\item
$\bfc = (1000000000)$,  $i = 3$, $\bft = (1010000000)$, $\bfq = (00100) = s_5$, so $k=5$. Compute $CntTry = cnt + \BFX(n-i+\beta-1,k)=236+\BFX(12,5)=308<m-1$, so set $cnt = 308$.
\item
$\bfc = (1010000000)$,  $i = 4$, $\bft = (1011000000)$, $\bfq = (01010)=s_{11}$, so $k = 11$.  Compute $CntTry = cnt + \BFX(n-i+\beta-1,k)=308+\BFX(11,11)=343<m-1$, so set $cnt = 343$.
\item
$\bfc = (1011000000)$, $i = 5$, $\bft = (1011100000)$ does not satisfy $(6,3)$-WWL constraint.
\item
$\bfc = (1011000000)$, $i = 6$, $\bft = (1011010000)$ does not satisfy $(6,3)$-WWL constraint.
\item
$\bfc = (1011000000)$, $i = 7$, $\bft = (1011001000)$, $\bfq = (11000)=s_{23}$, so $k=23$. Compute $CntTry = cnt + \BFX(n-i+\beta-1,k)=343+\BFX(8,23)=351<m-1$, so set $cnt = 351$.
\item
$\bfc = (1011001000)$, $i = 8$, $\bft = (1011001100)$ does not satisfy $(6,3)$-WWL constraint.
\item
$\bfc = (1011001000)$,  $i = 9$, $\bft = (1011001010)$, $\bfq = (00100)=s_5$, so $k=5$. Compute $CntTry = cnt + \BFX(n-i+\beta-1,k)=351+\BFX(6,5)=353>m-1$.
\item
$\bfc = (1011001000)$,  $i = 10$, $\bft = (1011001001)$, $\bfq = (01000)=s_9$, so $k=9$. Compute $CntTry = cnt + \BFX(n-i+\beta-1,k)=351+\BFX(5,9)=352=m-1$. Therefore, $\bfc = \bft = (1011001001)$ and $o(\bfc)=353$.
\end{itemize}
\end{example}

\begin{theorem} \label{thm: complexity encoder}
Algorithm~\ref{algr:encode} encodes a message $m\in\{1,\ldots,M_n\}$ to a $(\beta,p)$-WWL vector $\bfc\in\cS_n(\beta,p)$ such that $o(\bfc)=m$, and its time complexity and space complexity are both $O(n)$.
\end{theorem}
\begin{proof}

{\em Correctness}: The proof of the correctness of the encoder is similar to the proof of the correctness of the decoder. Therefore, we omit the details.

{\em Time complexity analysis}: It can be seen from the algorithm that the encoder scans the vector from left to right once and tries to set each entry to 1. Whenever the encoder sets an entry to 1, it first determines whether the constraint is satisfied. This takes $O(1)$ steps since we do not have to check the entire vector but only the $\beta$ bits to the left of the set entry. Then it uses binary search to find the corresponding prefix vector in $\BFX$, while the number of 1's is no more than $\frac{np}{\beta}$. Therefore, the complexity of the encoder is no more that $O(\frac{np}{\beta}\log M_{\beta-1})=O(\frac{np}{\beta}\log\sum_{i=0}^{p}\binom{\beta-1}{i})=O(n)$, where $p$ and $\beta$ are fixed numbers.

{\em Space complexity analysis}: The matrix $\BFX$ is the primary contributor to the space complexity. As is shown in the proof of Theorem~\ref{thm: complexity decoder}, the space complexity is also $O(n)$.
\end{proof}

Note that Algorithm~\ref{algr:encode} and Algorithm~\ref{algr:decode} establish a one-to-one mapping between $\{1,\ldots,M_n\}$ and $\cS_n(\beta,p)$. Therefore the rate of the encoder is maximized. If the blocklength goes to infinity, the rate of the encoder approaches $C_{WWL}(\beta,p)$.

\section{}\label{app:B}
In this section, we present the proof of Theorem~\ref{thm:cap 1 beta p}. The reason for which the proof of Theorem~\ref{thm:cap 1 beta p} is non-trivial is the following. Suppose the cell-level vector is updated from $\bfc_{i-1}$ to $\bfc_i$ on the $i$-th write. The encoder has full knowledge of $\bfc_{i-1}$ and $\bfc_i$ since we assume there is no noise in the updating procedure. The decoder is required to recover $\bfc_i + \bfc_{i-1}$ with full knowledge of $\bfc_i$ but zero knowledge of $\bfc_{i-1}$. This is similar to the work on memories with defects in~\cite{HG83}, where the most interesting scenario is when the defect locations are available to the encoder but not to the decoder. In general, it can be modeled as a channel with states~\cite{GK11} where the side information on states is only available to the encoder.

\begin{proof}
First we introduce some definitions. Recall that $\cS_n(\beta,p)$ is defined as the set all $(\beta,p)$-WWL vectors of length $n$. ${\cS}_n(\beta,p)$ will be written as ${\cS}$ for short if no confusion about the parameters can occur. Let $V_n=\{0,1\}^n$ be the $n$-dimensional binary vector space.
\begin{definition} \label{defn:good}
For a vector $\bfx\in V_n$ and a set $\cS\subset V_n$, we define $\cS+\bfx=\{\bfs+\bfx| \bfs \in \cS\}$ and denote it by $\cS(\bfx)$. We call vectors in $\cS(\bfx)$ \textbf {reachable} by $\bfx$ and we say $\cS(\bfx)$ is \textbf {centered} at $\bfx$.

For two subsets $B_1,B_2\subset V_n$, we define $B_1+B_2 = \{\bfb_1+\bfb_2| \bfb_1\in B_1, \bfb_2\in B_2\}$. We call a subset $B\subset V_n$ \textbf{$\cS$-good} if
\[
 \cS + B = \bigcup_{\bfb\in B}\cS(\bfb) = V_n,
 \]
i.e., $V_n$ is covered by the the union of translates of $\cS$ centered at vectors in $B$.
\end{definition}

\begin{lemma} \label{lem:good transfer}
If $B\subset V_n$ is $\cS$-good, then $\bft+B$ is $\cS$-good, $\forall \bft\in V_n$.
\end{lemma}
\begin{lemma} \label{lem:covering}
If $B\subset V_n$ is $\cS$-good, then $\forall \bfx \in V_n$, $\exists \bfb\in B, \exists \bfs\in \cS$, such that $\bfx+\bfs=\bfb$.
\end{lemma}

Lemma~\ref{lem:covering} guarantees that if $B\subset V_n$ is an $\cS$-good subset, then from any cell-state vector $\bfx$, there exists a $(\beta,p)$-WWL vector $\bfs$, such that $\bfx+\bfs\in B$. We skip the proofs of Lemma~\ref{lem:good transfer} and~\ref{lem:covering}, referring the reader to similar results and their proofs in~\cite{Cohen83}.

\begin{lemma} \label{lem:yielding constrained codes}
If $G_1,\ldots,G_{M}$ are pairwise disjoint $\cS$-good subsets of $V_n$, then there exists a $(1,\beta,p)$-constrained code of size $M$. In particular, if $G$ is an $\cS$-good $(n,k)$ linear code, then there exists a $(1,\beta,p)$-constrained code with rate $\frac{n-k}{n}$.
\end{lemma}
\begin{proof}
If $G_i$ is $\cS$-good for all $i\in[1:M]$, then from Lemma~\ref{lem:covering}, for any $\bfx\in V_n$ and $i\in[1:M]$, there exist $\bfg_i\in G_i$ and $\bfs_i\in\cS$, such that $\bfx + \bfs_i = \bfg_i$. Suppose the current cell-state vector is $\bfx$, then we can encode the message $i\in[1:M]$ as a vector $\cE(i,\bfx) = \bfx + \bfs_i\in G_i$, for some $\bfs_i\in \cS$. The decoder uses the mapping $\cD(\bfx) = i$, if $\bfx\in G_i$, to give an estimate of $i\in[1:M]$. This yields a $(1,\beta,p)$-constrained code of size $M$.

If $G_1,\ldots,G_{2^{n-k}}$ represent the cosets of an $\cS$-good $(n,k)$ linear code $G$, then each coset is $\cS$-good according to Lemma~\ref{lem:good transfer}. The rate of the resulting $(1,\beta,p)$-constrained code is $\frac{\log_2(2^{n-k})}{n}=\frac{n-k}{n}$.
\end{proof}

 Now we are ready to prove Theorem~\ref{thm:cap 1 beta p}.

Let $B_j$ be a randomly chosen $(n,j)$ linear code with $2^j$ codewords ($B_0=\{\boldsymbol{0}\}$), and let $m_{B_j}=|V_n/(B_j+\cS)|$ be the number of vectors not reachable from any vector in $B_j$. Let $\bfx\in V_n$ be a randomly chosen vector and let $Q_{B_j}$ be the probability that $\bfx\notin B_j+\cS$. Then we have
\[m_{B_j}=2^n Q_{B_j}.\]

The proof of the following lemma is based upon ideas discussed in~\cite[pp.~201-202]{Berger98}.
\begin{lemma}\label{lem:Q}
There exists a linear code $B_j$ such that
\[Q_{B_j}\leq Q_{B_0}^{2^j}.\]
\end{lemma}
\begin{proof}
Let $B_j=\{\bfy_1,\ldots,\bfy_{2^j}\}$ denote an $(n,j)$ linear code. If
\[S_{B_j} = B_j + \cS ,\]
then
\[ Q_{B_j} = 1-2^{-n}N_{B_j},\]
where $N_{B_j}=|S_{B_j}|$.

Let $\bfz \notin B_j$ and let $B_{j+1,\bfz}$ be the $(n,j+1)$ linear code formed by $(\bfz+B_j) \cup B_j$. It can be seen $B_{j+1,\bfz}$ comprises the $2^j$ vectors in $B_j$ plus $2^j$ new vectors of the form $\bfz+\bfy$, $\bfy\in B_j$. Let
\begin{align*}
S_{B_j,\bfz}^*& = \bfz + S_{B_j}
\end{align*}
It can be seen that $S_{B_j,\bfz}^*$ has the same cardinality as $S_{B_j}$. Therefore, it contains $N_{B_j}$ vectors, too, some of which may already belong to $S_{B_j}$. Since $S_{B_{j+1,\bfz}}=S_{B_j} \cup S_{B_j,\bfz}^*$, we have
\[
N_{B_{j+1,\bfz}}=2N_{B_j}-\big| S_{B_j} \cap S_{B_j,\bfz}^*\big|.
\]
Thus $N_{B_{j+1,\bfz}}$ is maximized by choosing $\bfz$ that minimizes $| S_{B_j} \cap S_{B_j,\bfz}^*|$.

Let us now calculate the average of $| S_{B_j} \cap S_{B_j,\bfz}^*|$ over all $\bfz\in V_n$. Here all $\bfz\in B_j$ are also considered since they will result in an overestimate of the average of $| S_{B_j} \cap S_{B_j,\bfz}^*|$. Then
\begin{align*}
\sum_{\bfz\in V_n} | S_{B_j} \cap S_{B_j,\bfz}^*| &= \sum_{\bfz\in V_n} \sum_{\bfx\in S_{B_j}}\textbf{1}_{\{\bfx\in S_{B_j,\bfz}^*\}} \\
&=\sum_{\bfx\in S_{B_j}}\sum_{\bfz\in V_n} \boldsymbol{1}_{\{\bfx\in S_{B_j,\bfz}^*\}}\\
&\stackrel{\textcircled{\scriptsize1}}{=}\sum_{\bfx\in S_{B_j}}\sum_{\bfz\in \bfx + S_{B_j}} 1\\
&\stackrel{\textcircled{\scriptsize2}}{=}\sum_{\bfx\in S_{B_j}}N_{B_j}\\
&= N_{B_j}^2,
\end{align*}
where $\boldsymbol{1}_A$ is the indicator function of the event $A$, i.e., $\boldsymbol{1}_A=1$ if $A$ is true and $\boldsymbol{1}_A=0$ otherwise.

Equality \textcircled{\scriptsize1} holds since, for a fixed $\bfx$, if $\bfz\in \bfx + S_{B_j}$, then $\bfx\in S_{B_j,\bfz}^*$ and vice versa. Equality \textcircled{\scriptsize2} holds since $|\bfx+S_{B_j}|=|S_{B_j}|=N_{B_j}$.
Thus, the average value of $| S_{B_j} \cap S_{B_j,\bfz}^*|$ is $2^{-n}N_{B_j}^2$. Since the minimum of $| S_{B_j} \cap S_{B_j,\bfz}^*|$ cannot exceed this average, we conclude that there exists $\bfz\in V_n$, such that $| S_{B_j} \cap S_{B_j,\bfz}^*|\leq N_{B_j}^2$. Then there exists $B_{j+1}$, such that
\[
N_{B_{j+1}} \geq 2N_{B_j}-2^nN_{B_j}^2.
\]
Thus,
\begin{align*}
Q_{B_{j+1}} &= 1-2^{-n}N_{B_{j+1}}\\
&\leq 1 - 2^{-n}(2N_{B_j}-2^{-n}N_{B_j}^2)\\
&=(1-2^{-n}N_{B_j})^2\\
&=Q_{B_j}^2.
\end{align*}
It follows that there exists $B_j$, such that $Q_{B_j}\leq Q_{B_0}^{2^j}$.
\end{proof}
\begin{lemma}\label{lem:good code exist}
If $j\geq n-\log |\cS| + \log n$, then there exists $B_j$ such that $m_{B_j}<1$.
\end{lemma}
\begin{proof}
Note that $Q_{B_0} = 1 - 2^{-n}\cdot N_{B_0} \leq 1-2^{-n}\cdot|\cS|$. Then there exists $B_j$, such that
\begin{align*}
Q_{B_j} &\leq Q_{B_0}^{2^j} \\
&\leq (1-2^{-n}|\cS|)^{2^j}\\
&\leq (1-2^{-n}|\cS|)^{2^{n-\log |\cS| + \log n}}\\
&= (1-2^{-n}|\cS|)^{2^{n}|\cS|^{-1}\cdot n}\\
&< e^{-n} < 2^{-n}.
\end{align*}
Then $m_{B_j}=2^n Q_{B_j} < 1$.
\end{proof}

Since $m_{B_j}$ is an integer and $m_{B_j}<1$, there exists an $(n,j)$ linear code $B_j$ such that $m_{B_j}=0$, i.e., an $\cS$-good $B_j$ exists. According to Lemma~\ref{lem:yielding constrained codes}, there exists a sequence of $(1,\beta,p)$-constrained codes of length $n$ and rate $R_n(1,\beta,p)$ such that
\begin{align*}
\sup_{n} R_n(1,\beta,p)&\geq \lim_{n\rightarrow \infty} \frac{{n-(n-\log |\cS| + \log n)}}{n}\\
&= \lim_{n\rightarrow \infty} \frac{\log |\cS| - \log n}{n}\\
&= \lim_{n\rightarrow \infty} \frac{\log |\cS|}{n}\\
&= C_{WWL}(\beta,p).
\end{align*}
We have seen in Theorem~\ref{th:upper} that $C(1,\beta,p)\leq C_{WWL}(\beta,p)$. This concludes the proof that $C(1,\beta,p)= C_{WWL}(\beta,p)$
\end{proof}

\bibliographystyle{IEEEtrans}
\bibliography{reference}
\end{document}